\documentclass[12pt]{article}
\usepackage[utf8]{inputenc}
\usepackage{newtxtext}
\usepackage[
  backend=biber,
  style=nejm,
  sorting=none,
  sortcites=true
]{biblatex}
\DeclareFieldFormat{labelnumberwidth}{#1\addperiod}
\usepackage{amsmath}
\usepackage{url}
\usepackage{enumerate}   
\usepackage{pgf,tikz}
\usepackage{enumitem}   
\usepackage[most]{tcolorbox}
\usepackage{pdflscape}
\usepackage{bbding}
\usepackage{prodint}
\usepackage{subcaption,booktabs}
\usepackage{tikz}
\usepackage{pdflscape}
\usepackage{multirow}
\usepackage{authblk}

\usetikzlibrary{arrows,shapes.arrows,shapes.geometric, shapes.multipart,decorations.pathmorphing,positioning, swigs}

\usepackage{footnote}

\usepackage{amsthm}
\usepackage{amssymb}
\usepackage{newtxmath}
\usepackage{float}

\usepackage{mathtools}
\usepackage{xcolor}
\usepackage{graphicx}

\usepackage{fullpage}
\usepackage[colorinlistoftodos]{todonotes}

\newcommand{\obs}{\mathrm{obs}}

\makeatletter
\renewcommand{\boxed}[1]{\text{\fboxsep=.2em\fbox{\m@th$\displaystyle#1$}}}
\makeatother

\usepackage{listings}
\usepackage{enumerate}
\usepackage{comment}

\usepackage{booktabs} 
\usepackage{subcaption}
\usepackage{hyperref}
\usepackage{etoolbox}
\usepackage{chngcntr}
\usepackage{cleveref}

\newtheorem{definition}{Definition}
\newtheorem{remark}{Remark}

\newtheorem{theorem}{Theorem}[section]
\newtheorem{proposition}{Proposition}[section]

\newtheorem{lemma}{Lemma}[section]

\makeatletter
\newcommand*\bigcdot{\mathpalette\bigcdot@{.5}}
\newcommand*\bigcdot@[2]{\mathbin{\vcenter{\hbox{\scalebox{#2}{$\m@th#1\bullet$}}}}}
\makeatother

\newcommand{\indep}{\perp \!\!\! \perp}

\newcommand{\probP}{\text{I\kern-0.15em P}}

\title{The causal interpretation of acceleration factors}
\author[1,2]{Mari Brathovde \thanks{Corresponding author: Mari Brathovde. mari.brathovde@medisin.uio.no}}
\author[3,4]{Hein Putter}
\author[2]{Morten Valberg}
\author[5,6]{Richard A.J. Post}
\affil[1]{Oslo Centre for Biostatistics and Epidemiology, Oslo University Hospital, Oslo, Norway}
\affil[2]{Oslo Centre for Biostatistics and Epidemiology, Department of Biostatistics, Institute of Basic Medical Sciences, University of Oslo, Oslo, Norway}
\affil[3]{Department of Biomedical Data Sciences, Leiden University Medical Center, Leiden, the Netherlands}
\affil[4]{Mathematical Institute, Leiden University, Leiden, the Netherlands}
\affil[5]{Department of Biostatistics, Erasmus Medical Center, University Medical Center Rotterdam, Rotterdam, the Netherlands}
\affil[6]{Department of Epidemiology, Erasmus Medical Center,
University Medical Center Rotterdam, Rotterdam, the Netherlands}
\begin{document}
\maketitle
\begin{abstract}
In studies of time-to-event outcomes with unmeasured heterogeneity, the hazard ratio for treatment is known to have a complex causal interpretation. Accelerated failure time (AFT) models, which assess the effect on the survival time ratio scale, are often suggested as a better alternative because they model a parameter with a direct causal interpretation while allowing straightforward adjustment for measured confounders. In this work, we formalize the causal interpretation of the acceleration factor in AFT models for data under independent censoring. We prove that the acceleration factor is a valid causal effect measure, even in the presence of frailty and treatment effect heterogeneity. Through simulations from structural causal models, we show that the acceleration factor better captures the causal effect than the hazard ratio when both AFT and conditional proportional hazards models apply. Additionally, we extend the interpretation to systems with time-dependent acceleration factors, illustrating the impossibility of distinguishing between a time-varying homogeneous effect and unmeasured effect heterogeneity. While the causal interpretation of acceleration factors is promising, we caution practitioners about potential challenges for the interpretation in the presence of effect heterogeneity.
\end{abstract}

\vspace{1em}  
\vspace{1em}

\section{Introduction}
It is widely recognized that even for randomized controlled trials (RCTs), estimands on the hazard rate scale, like the hazard ratio, may not be well-suited as causal estimands.\supercite{hernan2010hazards, post2024built, aalen2015does, martinussen2020subtleties, Fay2024, Dumas2025} The issue is that exchangeability can be lost over time in settings with unmeasured heterogeneity due to the inherent conditioning on survival when considering the hazard scale, i.e., the so-called built-in selection bias. The hazard ratio is affected by both the difference in characteristics of treated and untreated survivors and the conditional (on these characteristics) treatment effect. Therefore, to achieve interpretable causal estimands, it is advisable to use effect measures that are not affected by the built-in survival selection. For instance, the survival function is free from selection as it does not require conditioning on survival and will at any time $t$ concern the entire population rather than a subpopulation of survivors. Accordingly, estimands such as contrasts of survival functions or restricted mean survival times (RMST) are often suggested as favorable alternatives that have a straightforward causal interpretation.\supercite{hernan2010hazards, post2024built} Another alternative involves using accelerated failure time (AFT) models.\supercite{hernan2010hazards, aalen2015does} Unlike the commonly used Cox proportional hazards model, which assesses effects on the hazard scale, the AFT model measures effects on the survival time ratio scale. In particular, in AFT models, parameters act to accelerate (or decelerate) event times relative to a baseline time scale. 
Specifically, the ratio of survival times under treatment and no treatment (and, consequently, their quantiles) equals the acceleration factor $\theta$, which can be extended to a time-varying acceleration factor $\theta(t)$ allowing the ratio of quantiles to vary over $t$, as described in Chapter 5 of Cox and Oakes and related work.\supercite{cox1984analysis, robins1992semiparametric, Wei1992}

To control for confounders, survival curves can be adjusted directly using inverse probability of treatment (IPT) weights after estimation of the propensity scores,\supercite{Xie2005} or they can be adjusted by use of covariate adjustment in a time-to-event model. The latter is preferable when studying effect modification or when the relations of the confounders on the outcome are better known than the relations on the treatment assignment. Proportional hazards models can be used to adjust for covariates, but for valid causal inference, it will be necessary to subsequently derive the survival function from the fitted parameters on the hazard scale because of the complex causal interpretation of hazard ratios.\supercite{hernan2010hazards} 
On the contrary, the parameters of AFT models still concern the survival time ratio and may, therefore, be directly used to quantify the causal effect. AFT models can thus offer a more straightforward way to follow the general advice of evaluating causal effects through survival curves. By committing to a single estimand throughout the analysis, AFT models yield a workflow that is accessible to practitioners. 

The focus of this paper is the formalization of the causal interpretation of the estimands of AFT models in the presence of unobserved heterogeneity, both in the hazard rate under no treatment (i.e., frailty) and in treatment effect. We prove that the acceleration factor indeed yields an appropriate causal interpretation in the presence of frailty and heterogeneity in the treatment effect and present examples using structural causal models. The fact that frailty does not affect the causal interpretation in AFT models has been pointed out previously in the time-invariant setting without effect heterogeneity,\supercite{keiding1997role, aalen2015does} but is here formalized. This formalization highlights a key distinction between the acceleration factor and the estimand of the Cox model. In the presence of frailty, the latter estimand differs from the conditional (on frailty) causal effect,\supercite{post2024built} whereas the acceleration factor preserves the desired causal interpretation. Moreover, our results apply to the general settings with time-varying effects and effect heterogeneity. In these settings, the acceleration factor is time-varying and again has a clear and meaningful causal interpretation. 


\section{Notation and framework}
Let $T_i$ and $A_i$ denote the factual time-to-event outcome and treatment assignment of individual $i$. $T^a_i$ is the potential outcome if individual $i$ had been assigned to treatment $a$. We use latent variables to parameterize heterogeneity in the potential outcome distributions. Specifically, heterogeneity in $T_i^0$ is parameterized using a random variable $U_{0i}$, which represents the frailty of individual $i$. 
Furthermore, heterogeneity in the effect of treatment on the survival time, i.e. the relative rate at which $T^{a}$ progresses compared to $T^{0}$, is parameterized by the random variable $U_{1}$. 

The decomposition into $U_0$ and $U_1$ should not be interpreted as a separation into distinct biological mechanisms. Rather, it is a parameterization of latent heterogeneity where the role of $U_0$ is to collect the latent variation in baseline risk (frailty), whereas $U_1$ is to collect the latent variation relevant for the relative acceleration of $T^a$ compared with $T^0$ (treatment effect heterogeneity). The variables $U_0$ and $U_1$ may be multidimensional and their joint distribution is unrestricted. In particular, they may be statistically dependent, and latent factors that are associated with both baseline prognosis and relative treatment effects may be represented through both components. 

Given $U_{0}=u_{0}$ and $U_{1}=u_{1}$, the conditional causal effect of treatment may be characterized by linking the quantiles of the conditional $T^{a}$ and $T^{0}$ distributions. 

\begin{definition}[Conditional causal acceleration factor]
    For all $u_0$ in the support of $U_0$, the conditional causal acceleration factor for treatment level $a$ equals 
\begin{equation}\label{explicit_form1}
\theta_c(u_1,a,t)
\coloneqq
\frac{1}{t}\,
S^{-}_{T^0 \mid U_0=u_0}
\!\left(
S_{T^a \mid U_0=u_0,\,U_1=u_1}(t)
\right),
\end{equation}
where
$S_{T}(t) = \mathbb{P}(T>t)$ and $
S^{-}_{T^0 \mid U_0=u_0}(p)
\coloneqq
\inf\{s\ge0 : S_{T^0 \mid U_0=u_0}(s)\leq p\}$ is the generalized left inverse.
\end{definition}

Thus, for all $u_0$ in the support of $U_0$, the conditional causal acceleration factor satisfies 
\begin{align}\label{quantile_relation}
S_{T^a \mid U_0=u_0,\,U_1=u_1}(t)
=
S_{T^0 \mid U_0=u_0}\!\left(t\,\theta_c(u_1,a,t)\right),
\end{align}
so that $\theta_{c}$ relates the quantiles of the potential outcome distributions and has the interpretation that the $1-S_{T^a|U_{0} = u_0,U_{1} = u_1}(t)$ quantile of $T^0 \,| \,U_0 = u_0$ is $\theta_{c}(u_{1}, a, t) \, t$.\supercite{pang2021flexible} 
When $U_0 \indep U_1$, the interpretation of the conditional acceleration factor $\theta_c$ also applies when $U_0$ is marginalized out,
\begin{align} \label{coll-u0}
    S_{T^a | U_{1} = u_1}(t) =  
    \int S_{T^0 | U_{0} = u_0}(t \, \theta_{c}(u_{1}, a, t)) \, dF_{U_0 = u_0 | U_1 = u_1}(u_0) =  S_{T^0}(t \, \theta_{c}(u_{1}, a, t)).
\end{align}

\begin{remark}
As an alternative to \cref{quantile_relation}, one can consider

\begin{equation}
\theta_c^{\ast}(u_1,a,s) \coloneqq \frac{1}{s}
S^{-}_{T^a \mid U_0=u_0,\,U_1=u_1}
\!\left(
S_{T^0 \mid U_0=u_0}(s)
\right),
\end{equation}
 which has the interpretation that the $1-S_{T^0|U_{0} = u_0}(s)$ quantile of $T^a \,| \,U_0 = u_0, U_1 = u_1$ is $\theta^{\ast}_{c}(u_{1}, a, s) \, s$. Moving to the quantile scale,
%
\begin{equation}
    \theta_c (u_{1}, a, Q_{T^a | U_0 = u_0, U_1 = u_1}(p)) = 1/\theta_c^{\ast} (u_{1}, a, Q_{T^0 | U_0 = u_0}(p)), 
\end{equation} where $Q$ is the quantile function, $Q_{T}(p) = \inf \{ t \geq 0 \colon  S_{T}(t) \leq 1-p \}$.
\end{remark}

Analogously to the conditional causal acceleration factor $\theta_c$ (\ref{quantile_relation}), the marginal causal acceleration factor $\theta$ can be defined so that the $1-S_{T^a}(t)$ quantile of $T^0$ equals $\theta(a, t)t$.
\begin{definition}[Causal acceleration factor]\label{def:caf}
The causal acceleration factor for treatment level $a$ equals
\begin{align} \label{acc:factor}
    \theta(a,t) \coloneqq \frac{1 }{t} \left( S_{T^0}^{-} \left( S_{T^a}(t) \right) \right).
\end{align}
\end{definition}
Several numerical examples of causal acceleration factors are provided in Section \ref{sec:4}. Throughout, we denote by \(\theta(a)\) the causal acceleration factor for treatment \(a\), which in general may vary over time \(t\). For brevity, we omit the explicit dependence on \(t\) when referring to the estimand $\theta(a, \cdot)$ itself, but include it when discussing its value at a specific time point. Cases where \(\theta(a,t)\) is constant in \(t\) are explicitly noted. 
We will present plots of $p$ against $\theta(a, Q_{T^a}(p))$ in this paper where $Q_{T^0} (p) = Q_{T^a}(p) \theta(a, Q_{T^a}(p))$ so that $\theta(a,Q_{T^a}(p))$ equals the ratio of quantiles of the $T^0$ and $T^a$ distributions. 


It is worth noting that, in addition to $\theta$, the following related estimand is sometimes referred to as the time-varying acceleration factor:\supercite{crowther2023flexible}
\begin{align} \label{def:eta}
    \eta(a, t) &\coloneqq \frac{\partial}{ \partial t} \left( S_{T^0}^{-1} \left( S_{T^a}(t) \right) \right).
\end{align} In the case of time-invariant and homogeneous conditional causal acceleration factors, i.e. $\theta_{c}(u_{1}, a, t) = \theta_{c}(a)$, the estimands $\theta$ and $\eta$ are equal. 

Since the causal acceleration factor links the quantiles of the distributions of potential outcomes, contrasts of expectations of $T^a$ or $\log T^a$ are identified by $S_{T^0}$ and $\theta$ as shown in \cref{lemma:estimands}.
\begin{lemma} \label{lemma:estimands}

For 
$\theta$ as introduced in \cref{def:caf}, 
\begin{align}
   \mathbb{E}[\log T^a] - \mathbb{E}[\log T^0] &= \int^{\infty}_0 \log (t) \, d F_{T^0}\left( t \, \theta(a,t) \right) - \int^{\infty}_0 \log (t) \, d F_{T^0}(t),  \label{diff}\\
    \frac{\mathbb{E}[T^a]}{\mathbb{E}[T^0]} &= \frac{\int^{\infty}_0 S_{T^0} \left( t \, \theta(a,t) \right) \, dt}{\int^{\infty}_0 S_{T^0}(t) \, dt} \label{ratio}.
\end{align}
\end{lemma} 
In the examples presented in this paper, $\mathop{\mathbb{E}}[\log T^a] - \mathop{\mathbb{E}}[\log T^0]$ and $\frac{\mathbb{E}[T^a]}{\mathbb{E}[T^0]}$ are held constant across settings, while the underlying causal acceleration factor $\theta$ differs.

In practice, it is only possible to relate the quantiles of the distributions of observed event times among exposed ($T |\, A=1$) and non-exposed ($T |\,  A=0$). To do so, we define the observed acceleration factor $\theta_{\obs}$.
\begin{definition}[Observed acceleration factor]\label{def:sec}
The observed acceleration factor for treatment level $a$ equals
\begin{align} \label{eq:obs_acc_factor}
        \theta_{\obs}(a,t) \coloneqq \frac{1 }{ t} \left(S^{-}_{T | A = 0} \left( S_{T | A = a} (t ) \right) \right).
\end{align}
\end{definition} The key question addressed in this work is how \(\theta_{\obs}\) relates to \(\theta\), and consequently, how \(\theta_{\obs}\) should be interpreted causally.

\section{Identifiability results}
In this section, we show that under causal consistency and in the absence of confounding, even in the presence of both frailty $U_0$ and an individual effect modifier $U_1$, the observed acceleration factor $\theta_{\obs}$ has a clear causal interpretation as it equals the causal acceleration factor $\theta$. 
Additionally, we provide the result identifying the causal acceleration factor from censored data under the assumption of independent censoring.
\begin{theorem} \label{prop:1}
If $T^{a} \indep A$ (exchangeability) and $T^{A} = T$ (causal consistency), then 
\begin{align*}
    \frac{1}{ t} \left( S_{T\mid A=0}^{-} \left( S_{T\mid A=a}(t) \right) \right) &=  \frac{1}{ t} \left( S_{T^0}^{-} \left( S_{T^a}(t) \right) \right), \qquad  \mathrm{for \, all \,} t,
\end{align*}
i.e. the observed acceleration factor $\theta_{\obs}$ and the causal acceleration factor $\theta$ are equal irrespective of the presence of both frailty $U_0$ and effect heterogeneity $U_1$.
\end{theorem}

Besides demonstrating the identifiability of $\theta$ from observed data, this result highlights a key distinction between the observed acceleration factor and both the observed hazard ratio and the observed hazard difference. The latter estimands are defined on the hazard scale, which introduces selection due to the dependence between treatment assignment and $U_0$ and $U_1$ introduced when conditioning on survival.\supercite{post2024bias, post2024built} In contrast, the AFT estimand, being on the survival scale, avoids conditioning on survival and is therefore not subject to selection bias. Specifically, the observed hazard ratio at time $t$,
\begin{align} \label{eq:cox_estimand}
\frac{\lim_{h \to 0} h^{-1} \mathbb{P} \left( T \in [t,t+h) \, | \, T \geq t, A = a\right)}{\lim_{h \to 0} h^{-1} \mathbb{P} \left( T \in [t,t+h) \, | \, T \geq t, A=0\right)},
\end{align}
has a complicated causal interpretation and only corresponds to the conditional (causal) hazard ratio of interest when the effect on the hazard scale is multiplicative and when there is no frailty and no effect heterogeneity.\supercite{post2024built} The observed hazard difference,
\begin{align*}
\lim_{h \to 0} h^{-1} \mathbb{P} \left( T \in [t,t+h) \, | \, T \geq t, A = a\right) - \lim_{h \to 0} h^{-1} \mathbb{P} \left( T \in [t,t+h) \, | \, T \geq t, A = 0\right),
\end{align*}
matches the conditional (causal) hazard difference of interest when the effect on the hazard scale is additive, in the presence of frailty but only in the absence of effect heterogeneity.\supercite{post2024bias} In contrast,  $\theta_{\obs}$ maintains the desired causal interpretation even when both $U_0$ and $U_1$ follow a non-degenerate distribution. The fact that frailty does not induce selection bias in AFT models has been pointed out previously for time-invariant acceleration factors\supercite{keiding1997role, aalen2015does} but is formalized and extended to include effect heterogeneity and the time-varying setting in \cref{prop:1}.

Most time-to-event observations are subject to censoring. For each individual $i$, we observe a (possibly) right-censored survival time $\widetilde{T}_i$ together with an indicator $D_i$ that takes the value $1$ when $\widetilde{T}_i = T_i$  and the value $0$ when $\widetilde{T}_i < T_i$. Under the assumption of independent right censoring, 
\begin{align} \label{cond:indep_cens}
 \lim_{h \to 0} \frac{1}{h} \mathbb{P} \left( \tilde{T} \in [t, t+h), D = 1 \, | \, \tilde{T} \geq t, A = a \right) &=  \lim_{h \to 0} \frac{1}{h} \mathbb{P} \left( T \in [t, t+h) \, | \, T \geq t, A = a \right), \,\, \forall \, a,
\end{align} the causal acceleration factor is also identified by the censored data. 
\begin{proposition} \label{prop:kaplan-meier}
Under $T^{a} \indep A$ (exchangeability), $T^{A} = T$ (causal consistency) and the independent censoring assumption \cref{cond:indep_cens},
\begin{align*}
    \theta(a,t) &= \frac{1}{t} \left(S_{T | A = 0}^{-} \left( S_{T | A = a} \left(t \right)\right)  \right),
\end{align*}
for all $t$ in the support of the censored survival times, where
\begin{align*}
    S_{T | A = a}(t) &= \exp \left( - \int_0^t  \lim_{h \to 0} \frac{1}{h} \mathbb{P} \left( \tilde{T} \in [u, u+h), D = 1 \, | \, \tilde{T} \geq u, A = a \right) \, du \right).
\end{align*}
\end{proposition}

In the remainder of this paper, we focus on a binary treatment 
$a$. Consequently, the argument 
$a$ can be omitted from the estimands, simplifying the notation as $\theta_c(U_1, a, t) = \theta_c(U_1, t)$, $\theta( a, t) = \theta(t)$ and $\theta_{\obs}(a, t) = \theta_{\obs}(t)$. 

\section{Acceleration factors in the presence of unmeasured heterogeneity}\label{sec:4}
In the next two subsections, we study settings where the counterfactual lifetimes are decreased or increased by a time-invariant rate, $\theta_c(U_1, t) \equiv \exp f_1(U_1, 1)$, so that potential event times satisfy 
\begin{align} \label{eq:dist}
S_{T^1\mid U_0,U_1}(t)
&=
S_{T^0\mid U_0}\left( t\exp f_1(U_1,1)\right).
\end{align}

For simulation, we implement one convenient coupling of these conditional distributions using structural causal AFT models. In particular, the simulations use shared latent random variables when generating $T_i^0$ and $T_i^1$. This coupling improves simulation efficiency, but the estimands considered below depend only on the induced conditional and marginal potential-outcome distributions, not on the particular joint coupling of $T_i^0$ and $T_i^1$.

All code used in the examples presented
in this section can be found at \url{https://github.com/marbrath/causal_AFT}.

\subsection{Effect homogeneity} \label{sec:homo}
We first consider systems with homogeneous causal effects and frailty. In this setting $U_1$ is degenerate, equivalently $f_1(U_1, 1) = f_1(1)$, so that \cref{eq:dist} reduces to
\begin{align*}
    S_{T^1 \mid U_0} (t) &= S_{T^0 \mid U_0} (t \exp f_1(1)).
\end{align*}
This is the setting that is typically considered when discussing the causal interpretation of the observed hazard ratio.\supercite{hernan2010hazards, aalen2015does}

Note that in this case,
\begin{gather*}
S_{T^1}(t)
= \mathbb{E}[S_{T^1\mid U_0}(t)] 
= \mathbb{E}[S_{T^0\mid U_0}\left(t\exp f_1(1)\right)] 
= S_{T^0}\left(t\exp f_1(1)\right).
\end{gather*}
Thus, the marginal acceleration factor equals the conditional acceleration factor,
$\theta\equiv\theta_c\equiv\exp f_1(1)$, and $\eta$ coincides with these
quantities. In this sense, the acceleration factor is collapsible over $U_0$
in the absence of effect heterogeneity.

This argument parallels the approach of Daniel et al.,\supercite{daniel2021making} who characterize collapsibility through the so-called characteristic collapsibility function (CCF). In the present AFT setting, treatment acts by multiplying with a constant $1/\exp f_1(1)$, so the corresponding CCF is simply the time scaling map $t \mapsto t / \exp f_1(1)$. Because this map is linear, it commutes with marginalization over $U_0$, yielding collapsibility. Due to collapsibility, the observed acceleration factor \(\theta_{\obs}\) admits both a conditional interpretation, as the ratio of quantiles of the conditional potential outcome distributions, and a population-level interpretation, as the ratio of quantiles of the corresponding marginal distributions.

By \cref{lemma:estimands}, the constant acceleration factor also has the
following interpretations:
\begin{align}
    \mathbb{E}[\log T^1] - \mathbb{E}[\log T^0]
    &= \log(1/\theta), \label{eq:est1} \\
    \frac{\mathbb{E}[T^1]}{\mathbb{E}[T^0]}
    &= 1/\theta. \label{eq:est2}
\end{align}
Hence, the causal effect is constant across all quantiles of the potential
outcome distributions, and the observed acceleration factor $\theta_{\obs}$
admits both a conditional interpretation, as the ratio of quantiles of the
conditional potential outcome distributions, and a marginal interpretation, as
the ratio of quantiles of the corresponding marginal potential outcome
distributions.

For the simulations, we implement this setting using the SCM presented in \cref{SCAFT-effect-homo}.
\begin{align} \label{SCAFT-effect-homo}
A_i &\coloneqq f_A(N_{A_i}), \nonumber\\
T^0_i &\coloneqq \inf \left\{ t > 0 \,\colon
\exp\!\left(-\int_0^t f_0(U_{0i}, s)\, ds \right)
\leq N_{T_i} \right\}, \nonumber\\
 T^1_i &\coloneqq T^0_i / \exp f_1(1). 
\end{align}
Here, $f_A$ denotes the generalized inverse of the distribution function of $A$, $f_0(U_0,t)$ is the conditional hazard rate of $T^0$ at time $t$ given the frailty $U_0$, and $f_1(1)$ is the log survival time ratio. Moreover, $N_{A_i}, N_{T_i} \sim \mathrm{Unif}[0,1]$, $U_{0i} \sim F_{U_0}$ and $U_{0i} \indep N_{T_i}$.  In the numerical examples below,
$F_{U_0}$ is specified as either a Gamma or inverse Gaussian frailty
distribution.

In the absence of censoring, one can employ the above formulations of the acceleration factor. We will now consider the formulation of the acceleration factor given by \cref{eq:est1}, and thus also the AFT model on the log scale, 
\begin{align*} 
     \log T^1_i &\coloneqq \log T^0_i - f_1(1). 
\end{align*}
If $f_0(U_{0i}, t) = \frac{\kappa}{\sigma}t^{\frac{1}{\sigma} -1 }U_{0i}$ and $f_1(1) = \beta \sigma$, then 
\begin{align} \label{eq:AFT_logform}
    \log T^1_i &= - \log (\kappa) \, \sigma - \log (U_{0i}) \, \sigma - \beta \, \sigma \,  + \sigma \, W_i,
\end{align}
where $W_i$ is extreme-value distributed with location $0$ and scale $1$.

It can be shown that \cref{eq:AFT_logform} can be reformulated as a Weibull (conditional) proportional hazards model, as described in Chapter 5 of Cox and Oakes,\supercite{cox1984analysis}
\begin{align} \label{eq:weibulll}
    \lambda^1_i(t) &= \frac{\kappa}{\sigma} t^{\frac{1}{\sigma}-1}  U_{0i} \exp(\beta),
\end{align} where, given $U_{0i}=u_{0}$, $\lambda^{1}_{i}(t)$ equals $\lim_{h\rightarrow 0} h^{-1}\mathbb{P}\left(T^{1} \in [t,t+h) \mid T^{1} \geq t, U_{0i}=u_{0} \right)$, the conditional hazard rate of the potential outcome of individual $i$ under treatment $a=1$. 

As shown by Post et al.,\supercite{post2024built} when the proportional hazards assumption holds conditionally on $U_0$ (\ref{eq:weibulll}), the observed hazard ratio deviates from the conditional (causal) hazard ratio $\exp(\beta)$ and is time-varying. Hence, the proportional hazards assumption does not hold marginally, and the estimand of the misspecified marginal (as $U_0$ is not observed) Cox proportional hazards model is the average of the logarithm of observed hazard ratios weighted by the observed event times and is therefore affected by the censoring distribution.\supercite{post2024built}

To illustrate that Theorem \ref{prop:1} applies, we derive the value of $\theta_{\obs}$ on simulated data where $T | A, U_0$ satisfies \cref{eq:weibulll}. $A$ is randomly assigned, $\mathbb{P}( A = 1) = 0.5$ and $N_A \indep (U_0, W)$. Let $\sigma = 1/3, \kappa = 1/60$, and suppose that $U_{0}$ follows either a Gamma distribution, $U_0 \sim \Gamma(\rho^{-1}_0, \rho_0)$, or inverse Gaussian distribution, $U_0 \sim \mathrm{IG}(1, \rho^{-1}_0)$, distribution so that $\mathbb{E}[U_0] = 1, \mathrm{var}(U_0) = \rho_0$. Since $S_{T^{1}}(t) = S_{T^{0}}(t \exp(\beta)^{1/3})$, the relation between $\theta$ and the conditional (causal) hazard ratio is $\theta = \exp(\beta)^{1/3}$. We first consider the setting without censoring and use the fact that $\theta_{\obs} = \exp \left( -\left( \mathbb{E} [\log T |A=1] - \mathbb{E}[\log T | A=0] \right) \right)$ to empirically estimate $\theta_{\obs}$ in each simulation of a large sample ($n_{\mathrm{obs}}=10^5$). 

For comparison, we also obtained the Cox estimate $\hat{\gamma}$ for each simulated data set by fitting the marginal Cox model, $\lambda(t \mid A) = \lambda_0(t) \exp (\gamma A )$, using the \texttt{coxph} function from the \texttt{R} package \texttt{survival}. While indeed $\theta_{\obs}$ equals $\theta$ in the presence of $U_0$ (cf. \cref{prop:1}), the Cox estimand deviates from the conditional (causal) hazard ratio so that the Cox estimator is a biased estimator for $\exp(\beta)$. The latter bias is affected by the frailty distribution and increases with increasing frailty variance as shown in \cref{table:1}. For Gamma frailty with \(\exp(\beta) = 1/3 \), the mean of the exponential of the 1000 Cox estimates equals \(0.477\), \(0.577\), and \(0.695\) for \(\rho_{0} = 0.5, 1,\) and \(2\), respectively. When \(\exp(\beta) = 3\), the corresponding means are \(2.095\), \(1.732\), and \(1.439\).

In the presence of censoring, the formulation of $\theta_{\obs}$ on the survival scale (\cref{eq:obs_acc_factor}) applies and is affected by neither $U_0$ nor the censoring distribution. While this result is immediate from Theorem~\ref{prop:1}, it is worth noting since the Cox estimator, in contrast, is influenced by both. In \cref{fig1-cox}, this is demonstrated using simulated data with $n_{\mathrm{obs}}=10^6$, Gamma-distributed frailty $U_0$ with $\mathrm{var}(U_0)=1$, and $\theta=0.693$. We consider varying follow-up times, expressed as quantiles of the $T^1$ distribution, and independent censoring given by an exponentially distributed censoring time $T_{C}$ with varying means.

For all $t \in (0, t_{\mathrm{FU}}]$, the expected value of $\hat{\theta}_{\obs}(t)$ does not depend on the follow-up time $t_{\mathrm{FU}}$ or the censoring distribution and equals $\theta$. In contrast, the expected value of the Cox estimator deviates from the conditional (causal) hazard ratio $\exp(\beta)=1/3$, and this deviation increases with longer follow-up and, on average, longer censoring times. This occurs because, in the presence of frailty, the observed hazard ratio varies over time due to selection among survivors, and the Cox estimator averages these time-varying observed hazard ratios with weights determined by the observed event-time distribution. Shorter follow-up times and heavier censoring reduce the contribution of later event times, where the impact of selection is larger. For example, for studies with follow-up time equal to $6.46$, corresponding to $F_{T^1}(6.46)=0.6$, the exponential of the Cox estimate equals $0.49$ in the absence of censoring and $0.45$ for exponentially distributed censoring times with mean $0.5\mathbb{E}[T^1]=3.41$, both deviating from the conditional (causal) hazard ratio $\exp(\beta)=1/3$.

\begin{figure}[htbp]
\centering
\includegraphics[width=\linewidth]{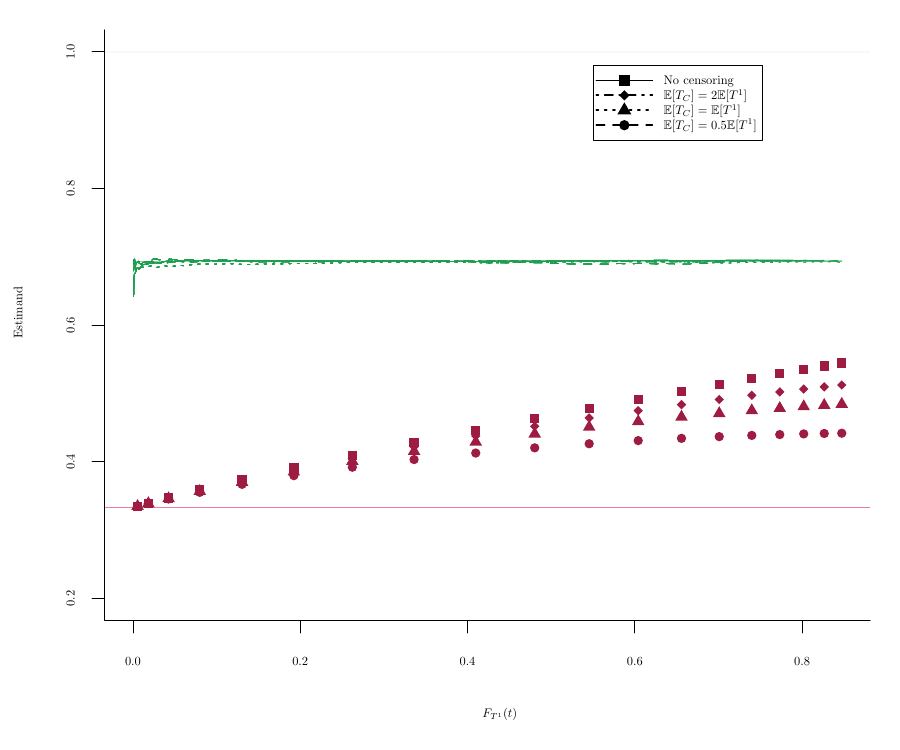}
\caption{
Cox estimates $\exp(\hat{\gamma})$ from the marginal Cox model (purple markers) evaluated at follow-up times $t_{\mathrm{FU}}$, and $\hat{\theta}_{\obs}$ (green curves) evaluated on $(0,t_{\mathrm{FU}}]$, where $t_{\mathrm{FU}}$ corresponds to increasing quantiles of the $T^1$ distribution ($x$-axis). Different marker shapes and line types represent different censoring regimes: no censoring, and exponential censoring times $T_{C}$ with means $2\mathbb{E}[T^1]$, $\mathbb{E}[T^1]$, and $0.5\mathbb{E}[T^1]$. The horizontal pink line marks the conditional causal hazard ratio $\exp(\beta)=1/3$. The true acceleration factor is $\theta=\exp(\beta)^{1/3}$, its reference line coincides with the green $\hat{\theta}_{\obs}(t)$ curves.
}
\label{fig1-cox}
\end{figure}
\newpage

\subsection{Effect heterogeneity} \label{subsec:effect_hetero}
In the presence of effect heterogeneity \(U_1\), a time-varying \(\theta\) arises naturally because the marginal time scale required to match the mixture of stratum-specific survival curves changes with \(t\). Equivalently, the causal effect varies per quantile of the potential outcome distributions and cannot be summarized with a single value. Consequently, results derived under effect homogeneity no longer hold. The conditional (which now depends on the level of $U_{1}$) and marginal causal estimands no longer coincide, implying that the observed acceleration factor represents only a population-level causal effect.  Moreover, \(\theta\) summarizes the cumulative time-scaling of survival under treatment \(a\) up to time \(t\), whereas \(\eta\) captures the instantaneous rate of change at time \(t\) among those still at risk. While \cref{eq:est1} and \cref{eq:est2} are the same for the upcoming examples, \(\theta\) differs across them, indicating that a single-valued estimand can no longer represent the causal effect.

For the simulations, we implement this setting using the SCM presented in \cref{SCAFT_lin_rel}. 
\begin{align}  \label{SCAFT_lin_rel}
A_i &\coloneqq f_A(N_{A_i}), \nonumber\\
T^0_i &\coloneqq \inf \left\{ t > 0 \,\colon
\exp\!\left(-\int_0^t f_0(U_{0i}, s)\, ds \right)
\leq N_{T_i} \right\}, \nonumber\\
T^1_i &\coloneqq T^0_i /  \exp f_1(U_{1i}, 1). 
\end{align}
Here, $f_A$ denotes the generalized inverse of the distribution function of $A$, $f_0(U_0,t)$ is the conditional hazard rate of $T^0$ at time $t$ given $U_0$, and $f_1(U_1,1)$ is the conditional log survival time ratio. Moreover, $N_{A_i}, N_{T_i} \sim \mathrm{Unif}[0,1]$, $U_{0i}, U_{1i} \sim F_{U_0,U_1}$ and $(U_{0i},U_{1i}) \indep N_{T_i}$. In the numerical examples below, the
joint distribution of $(U_0,U_1)$ is specified explicitly.

We extend the example of the previous subsection with effect heterogeneity, in particular $f_1(U_1, 1) = \log U_1$, $U_1 \indep U_0$, such that $T^1 = T^0 / U_1$. We consider a setting where $U_1$ equals $\mu_1$ ($<1$ for individuals that benefit) with probability $p_1$, $\mu_2$ ( $> 1$ for individuals that are harmed) with probability $p_2$ or $1$ (for individuals that are not affected). We refer to this distribution as the benefit-harm-neutral BHN($p_1$, $\mu_1$, $p_2$, $\mu_2$) distribution. Parameters $(p_1, \mu_1, p_2, \mu_2)$ such that $\mathbb{E}[U_1] \in \{(1/3)^{1/3}, 3^{1/3}\}$ and $\mathrm{var}(U_1) = \rho_1$ are found in Appendix B.4 of Post et al.\supercite{post2024built}

The evolution of $\theta$ over time is presented in \cref{fig:effect_hetero} (left) for $U_1 \sim \mathrm{BHN}(0.7, 0.3, 0.05, 5.10)$ ($\mathbb{E}[U_1] = (1/3)^{1/3}$) and $U_1 \sim \mathrm{BHN}(0.05, 0.5, 0.18, 3.53)$ ($\mathbb{E}[U_1] = 3^{1/3}$). The numerical values of the quantiles of $T^{1}$ can be found in Figure S7 (left). Since $U_1$ is independent of $U_0$, the  individuals with the highest (harming) value of $U_1$, i.e. $\mu_2$, will contribute more to the low quantiles of $T^1$, while those with lower (beneficial) values of $U_1$, i.e. $\mu_1$, will contribute more to the larger quantiles of $T^1$. Thus, for low quantiles, the acceleration factor will be closer to $\mu_2$, while for larger quantiles, the acceleration factor will be closer to $\mu_1$. As a reminder, $Q_{T^0} (p) = Q_{T^1}(p) \theta(Q_{T^1}(p))$, thus $\theta(Q_{T^1}(p))>1$ implies that the $p$-th quantile of $T^0$ is larger than that of $T^{1}$ so that a higher proportion of the population survives $Q_{T^1}(p)$ when no one is treated than when everyone is treated. Similarly, $\theta(Q_{T^1}(p))<1$ implies that when no one is treated, a smaller proportion of the population survives $Q_{T^1}(p)$ compared to when everyone would be treated. In the case of $\mathbb{E}[U_1] = (1/3)^{1/3}$, $\theta$ is greater than $1$ up until approximately the $0.1$ quantile, hence $S_{T^0}$ dominates $S_{T^1}$ in this area. At $\theta(t) = 1$, the curves cross such that $S_{T^1}$ starts dominating $S_{T^0}$ in the remaining quantile range (cf. Figure S7 (left)). Consequently, the proportion of individuals that survive $Q_{T^1}(p)$ is larger when treated for approximately $p > 0.1$. Analogously, for $\mathbb{E}[U_1] = 3^{1/3}$, $\theta$ slightly dips below $1$ shortly after the $0.8$ quantile, hence $S_{T^0}$ dominates $S_{T^1}$ in the majority of the quantile range and the proportion of individuals that survive $Q_{T^1}(p)$ is smaller when treated for $p < 0.8$ (cf. Figure S7 (left)).  For reference, the estimand $\mathbb{E}[U_1]$ is included in \cref{fig:effect_hetero}, which, in the case of effect homogeneity ($\rho_1 = 0$) equals the time-invariant marginal causal effects $\theta$, cf. \cref{eq:est2}.

\begin{figure}[htbp]
\centering
\include{effect_hetero}
\caption{$\theta$ when $T^1 = T^0 / U_1$, $\lambda^a_{i}(t) = (t^2/20) U_{0i} \exp (\beta a), U_0 \sim \Gamma(1,1)$ and $U_1$ follows a BHN distribution with $\rho_1 = 1$, $\mathbb{E}[U_1] = 3^{1/3}$ (($p_1$, $\mu_1$, $p_2$, $\mu_2$) = $(0.05, 0.5, 0.18, 3.53)$) (green) and $\mathbb{E}[U_1] = (1/3)^{1/3}$ (($p_1$, $\mu_1$, $p_2$, $\mu_2$) = $(0.7, 0.3, 0.05, 5.10)$) (orange) (left); when $T^0  \sim \mathrm{Weibull}(\Lambda,2), \Lambda \sim X / \Gamma(1 + 1/2), X $ is categorical ($\mathbb{P}(X = 1 ) = \mathbb{P}(X=10) = 0.5 $) and $U_1$ is as specified for left-hand side (right). See associated $S_{T^0}$ and $S_{T^1}$ in Figure S7. When $\theta(Q_{T^1}(p))>1$, the $p$-th quantile of $T^1$ is smaller than that of $T^0$ and when $\theta(Q_{T^1}(p))<1$, the opposite holds.}
\label{fig:effect_hetero}
\end{figure}

To demonstrate the dependence of $\theta$ on $S_{T^0}$, we also present an example where $T^0$ is Weibull mixture distributed, $\mathrm{Weibull}(\Lambda,2), \Lambda \sim X / \Gamma(1 + 1/2), X $ is categorical ($\mathbb{P}(X = 1 ) = \mathbb{P}(X=10) = 0.5 $), $U_1$ as defined above and again $U_1 \indep T^0$. Here $T^1$ is a finite mixture, in particular, 
$S_{T^1} (t) = \sum_{i=1}^{3} p_i S_{T^0} (t \mu_i)$ for $S_{T^0} (t) = 0.5 S_{T'_0} (t) + 0.5 S_{T''_0} (t)$, $T'_0 \sim \mathrm{Weibull}(1/\Gamma(1+1/2), 2), T''_0 \sim \mathrm{Weibull}(10/\Gamma(1+1/2), 2)$. The resulting $\theta$ is displayed in \cref{fig:effect_hetero} (right), and the times corresponding to the quantiles of $T^{1}$ can be found in Figure S7 (right). To help explain the behaviour of $\theta$, we will refer to Figure S8 and Figure S9, which display the survival functions and densities of the components in the mixture.

As seen in Figure S9, initially, individuals who are harmed by treatment ($\mu_2 = 5.10$ (orange), $3.53$ (green)) experience the event, then the first mode of unaffected individuals experience the event before the first mode of individuals who benefit from treatment experience the event ($\mu_1 = 0.3$ (orange), $0.5$ (green)). Consequently, the acceleration factor decreases monotonically up until approximately the $0.55$ quantile ($t= 6.45$ (orange), $2.5$ (green)). Thereafter, the second mode of  unaffected individuals forces the acceleration factor to increase. As the group that benefits reaches its second mode and the group of unaffected individuals becomes smaller relative to the group which benefits, the acceleration factor once again decreases.

Finally, we consider an example with a continuous (Gamma-distributed) $U_1$ to assess the effect of the variability of $U_1$ on $\theta$. In \cref{fig:effect_hetero2}, it is demonstrated that greater variance yields an increasingly heterogeneous relationship between the quantiles. The times corresponding to the quantiles of $T^{1}$ can now be found in Figure S10.

\begin{figure}[htbp]
\centering
\include{effecthetero-gammaU1}
\caption{$\theta$ when $T^1 = T^0 / U_1$, $T^0$ as specified in \cref{fig:effect_hetero}, $U_1$ follows a Gamma distribution with $\mathrm{Var}(U_1) = 0.5,1,2$ (dashed, solid, dotted), $\mathbb{E}[U_1] = 3^{1/3}$ (green) and  $\mathbb{E}[U_1] = (1/3)^{1/3}$ (orange). See the associated $S_{T^0}$, $S_{T^1}$ in Figure S10.}
\label{fig:effect_hetero2}
\end{figure}  

Consequently, it is demonstrated that time-invariant conditional causal effects can result in time-varying marginal causal effects when effect heterogeneity is present. Thus, a homogeneous time-varying causal effect cannot be distinguished from a time-invariant but heterogeneous causal effect. Across all examples presented in this section, the expectation of the conditional causal effect $\mathbb{E}[U_1]$ is fixed and equal to either $(1/3)^{1/3}$ or $3^{1/3}$, yet quite different behaviours of the related marginal causal estimands $\theta$ are observed. In the case of a time-varying marginal acceleration factor, the simple interpretation of the acceleration factor as a contrast of expected survival times, cf. \cref{eq:est2}, no longer applies.  
\newpage

\subsubsection{Case study: Federation Francophone de Cancerologie Digestive Group Study 9803}\label{S:casestudy}
Burzykowski\supercite{burzykowski2022semi} reviewed the properties of and estimation methods for the AFT model and presented a simulation study to discuss the applicability of the AFT model as an alternative to the proportional hazards model in the context of cancer clinical trials. As a practical example, the time-invariant semiparametric AFT model
\begin{align*}
    \log T = -\log (\theta_{\obs}) A + \varepsilon,
\end{align*}
where the error distribution $F_{\varepsilon}$ is left unspecified and
$A=1$ represents adding experimental chemotherapeutic agents to pre-existing regimens ($A=0$), was fitted separately to the progression-free survival times of each of the 20 trials in advanced gastric cancer. These trials were included in the individual-patient-data meta-analysis conducted by the GASTRIC Group,\supercite{Gastric2013} and the data were made publicly available as supplementary material by Buyse et al.\supercite{buyse2016statistical} To verify the appropriateness of the model, and thus the time-invariant effect, Burzykowski\supercite{burzykowski2022semi} compared the survival function of the residuals $\log T + \log (\theta_{\obs})  A$. A deviation was observed for the trial conducted by Bouch\'{e} et al.,\supercite{Bouche2004} which included $135$ participants for which the Kaplan-Meier curves can be found in Figure S5 in the Supplementary Material.

We will fit a time-varying AFT model to the Bouch\'{e} trial data, as represented in the GASTRIC/Buyse data set, to estimate $\theta_{\obs}$. Since the sample size is small, we cannot resort to a non-parametric estimator. Instead, we apply the flexible parametric method proposed by Crowther et al.\supercite{crowther2023flexible} as implemented in the \texttt{aft()} function from the \texttt{R} package \texttt{rstpm2}. Using a cubic spline with $3$ degrees of freedom for $\log(-\log S_{T\mid A=1})$ and a cubic spline with two pre-specified knots for $\log \theta_{\obs}$, the model fits quite well as illustrated in Figure S5 in the Supplementary Material. The default output of the $\texttt{aft()}$ function is in terms of $\eta(t)$ (cf. \cref{def:eta}), but we have written some additional code to output the estimate of $\theta_{\obs}$ and its uncertainty. The code for this case study can be found at \url{https://github.com/marbrath/causal_AFT}. The estimated $\theta_{\obs}$ is presented in \cref{fig:case_study}.

\begin{figure}[htbp]
\centering
\begin{subfigure}{0.5\textwidth}
  \centering
\begin{tikzpicture}[x=1pt,y=1pt]
\definecolor{fillColor}{RGB}{255,255,255}
\path[use as bounding box,fill=fillColor,fill opacity=0.00] (0,30) rectangle (267.40,237.40);
\begin{scope}
\path[clip] ( 49.20, 61.20) rectangle (242.20,218.20);
\definecolor{drawColor}{RGB}{0,0,0}

\path[draw=drawColor,line width= 0.8pt,line join=round,line cap=round] (235.05, 99.24) --
	(225.12, 98.73) --
	(215.20, 98.70) --
	(205.27, 98.93) --
	(195.34, 99.34) --
	(185.41, 99.91) --
	(175.48,100.63) --
	(165.56,101.51) --
	(155.63,102.56) --
	(145.70,103.81) --
	(135.77,105.30) --
	(125.84,107.06) --
	(115.92,109.20) --
	(105.99,111.83) --
	( 96.06,115.11) --
	( 86.13,119.17) --
	( 76.20,124.13) --
	( 66.28,130.42) --
	( 56.35,139.61);
\end{scope}
\begin{scope}
\path[clip] (  0.00,  0.00) rectangle (267.40,267.40);
\definecolor{drawColor}{RGB}{0,0,0}

\path[draw=drawColor,line width= 0.4pt,line join=round,line cap=round] ( 86.13, 61.20) -- (205.27, 61.20);

\path[draw=drawColor,line width= 0.4pt,line join=round,line cap=round] ( 86.13, 61.20) -- ( 86.13, 55.20);

\path[draw=drawColor,line width= 0.4pt,line join=round,line cap=round] (125.84, 61.20) -- (125.84, 55.20);

\path[draw=drawColor,line width= 0.4pt,line join=round,line cap=round] (165.56, 61.20) -- (165.56, 55.20);

\path[draw=drawColor,line width= 0.4pt,line join=round,line cap=round] (205.27, 61.20) -- (205.27, 55.20);

\node[text=drawColor,anchor=base,inner sep=0pt, outer sep=0pt, scale=  1.00] at ( 86.13, 39.60) {0.2};

\node[text=drawColor,anchor=base,inner sep=0pt, outer sep=0pt, scale=  1.00] at (125.84, 39.60) {0.4};

\node[text=drawColor,anchor=base,inner sep=0pt, outer sep=0pt, scale=  1.00] at (165.56, 39.60) {0.6};

\node[text=drawColor,anchor=base,inner sep=0pt, outer sep=0pt, scale=  1.00] at (205.27, 39.60) {0.8};

\path[draw=drawColor,line width= 0.4pt,line join=round,line cap=round] ( 49.20, 67.01) -- ( 49.20,196.23);

\path[draw=drawColor,line width= 0.4pt,line join=round,line cap=round] ( 49.20, 67.01) -- ( 43.20, 67.01);

\path[draw=drawColor,line width= 0.4pt,line join=round,line cap=round] ( 49.20, 99.32) -- ( 43.20, 99.32);

\path[draw=drawColor,line width= 0.4pt,line join=round,line cap=round] ( 49.20,131.62) -- ( 43.20,131.62);

\path[draw=drawColor,line width= 0.4pt,line join=round,line cap=round] ( 49.20,163.93) -- ( 43.20,163.93);

\path[draw=drawColor,line width= 0.4pt,line join=round,line cap=round] ( 49.20,196.23) -- ( 43.20,196.23);

\node[text=drawColor,rotate= 90.00,anchor=base,inner sep=0pt, outer sep=0pt, scale=  1.00] at ( 34.80, 67.01) {0.4};

\node[text=drawColor,rotate= 90.00,anchor=base,inner sep=0pt, outer sep=0pt, scale=  1.00] at ( 34.80, 99.32) {0.6};

\node[text=drawColor,rotate= 90.00,anchor=base,inner sep=0pt, outer sep=0pt, scale=  1.00] at ( 34.80,131.62) {0.8};

\node[text=drawColor,rotate= 90.00,anchor=base,inner sep=0pt, outer sep=0pt, scale=  1.00] at ( 34.80,163.93) {1.0};

\node[text=drawColor,rotate= 90.00,anchor=base,inner sep=0pt, outer sep=0pt, scale=  1.00] at ( 34.80,196.23) {1.2};

\path[draw=drawColor,line width= 0.4pt,line join=round,line cap=round] ( 49.20, 61.20) --
	(242.20, 61.20) --
	(242.20,218.20) --
	( 49.20,218.20) --
	cycle;
\end{scope}
\begin{scope}
\path[clip] (  0.00,  0.00) rectangle (267.40,267.40);
\definecolor{drawColor}{RGB}{0,0,0}

\node[text=drawColor,anchor=base,inner sep=0pt, outer sep=0pt, scale=  1.00] at (145.70, 15.60) {$F_{T^1}(t)$};

\node[text=drawColor,rotate= 90.00,anchor=base,inner sep=0pt, outer sep=0pt, scale=  1.00] at ( 10.80,139.70) {$\theta_{\text{obs}}(t)$};
\end{scope}
\begin{scope}
\path[clip] ( 49.20, 61.20) rectangle (242.20,218.20);
\definecolor{drawColor}{RGB}{0,0,0}

\path[draw=drawColor,line width= 0.4pt,dash pattern=on 4pt off 4pt ,line join=round,line cap=round] (235.05, 72.51) --
	(225.12, 75.05) --
	(215.20, 76.01) --
	(205.27, 76.38) --
	(195.34, 76.53) --
	(185.41, 76.62) --
	(175.48, 76.75) --
	(165.56, 76.96) --
	(155.63, 77.28) --
	(145.70, 77.70) --
	(135.77, 78.22) --
	(125.84, 78.82) --
	(115.92, 79.45) --
	(105.99, 80.02) --
	( 96.06, 80.39) --
	( 86.13, 80.49) --
	( 76.20, 80.62) --
	( 66.28, 81.75) --
	( 56.35, 87.14);

\path[draw=drawColor,line width= 0.4pt,dash pattern=on 4pt off 4pt ,line join=round,line cap=round] (235.05,136.17) --
	(225.12,130.12) --
	(215.20,128.39) --
	(205.27,128.35) --
	(195.34,129.17) --
	(185.41,130.50) --
	(175.48,132.18) --
	(165.56,134.14) --
	(155.63,136.39) --
	(145.70,138.98) --
	(135.77,142.03) --
	(125.84,145.74) --
	(115.92,150.44) --
	(105.99,156.68) --
	( 96.06,165.30) --
	( 86.13,177.02) --
	( 76.20,191.84) --
	( 66.28,208.95) --
	( 56.35,224.56);
\definecolor{drawColor}{RGB}{147,112,219}

\path[draw=drawColor,line width= 0.4pt,line join=round,line cap=round] (235.05, 90.22) --
	(225.12, 93.62) --
	(215.20, 96.01) --
	(205.27, 97.91) --
	(195.34, 99.52) --
	(185.41,100.95) --
	(175.48,102.24) --
	(165.56,103.45) --
	(155.63,104.62) --
	(145.70,105.77) --
	(135.77,106.96) --
	(125.84,108.26) --
	(115.92,109.73) --
	(105.99,111.40) --
	( 96.06,113.24) --
	( 86.13,115.02) --
	( 76.20,116.35) --
	( 66.28,117.13) --
	( 56.35,117.72);
\definecolor{drawColor}{RGB}{220,220,220}

\path[draw=drawColor,draw opacity=0.20,line width= 0.8pt,line join=round,line cap=round] ( 49.20,105.46) -- (242.20,105.46);
\end{scope}
\end{tikzpicture}
\end{subfigure}%
\begin{subfigure}{.5\textwidth}
  \centering
\begin{tikzpicture}[x=1pt,y=1pt]
\definecolor{fillColor}{RGB}{255,255,255}
\path[use as bounding box,fill=fillColor,fill opacity=0.00] (0,30) rectangle (267.40,237.40);
\begin{scope}
\path[clip] ( 49.20, 61.20) rectangle (242.20,218.20);
\definecolor{drawColor}{RGB}{0,0,0}

\path[draw=drawColor,line width= 0.8pt,line join=round,line cap=round] (235.05, 99.24) --
	(200.85, 98.73) --
	(179.76, 98.70) --
	(164.21, 98.93) --
	(151.74, 99.34) --
	(141.23, 99.91) --
	(132.10,100.63) --
	(123.96,101.51) --
	(116.57,102.56) --
	(109.77,103.81) --
	(103.42,105.30) --
	( 97.42,107.06) --
	( 91.68,109.20) --
	( 86.12,111.83) --
	( 80.64,115.11) --
	( 75.14,119.17) --
	( 69.48,124.13) --
	( 63.41,130.42) --
	( 56.35,139.61);
\end{scope}
\begin{scope}
\path[clip] (  0.00,  0.00) rectangle (267.40,267.40);
\definecolor{drawColor}{RGB}{0,0,0}

\path[draw=drawColor,line width= 0.4pt,line join=round,line cap=round] ( 82.62, 61.20) -- (204.17, 61.20);

\path[draw=drawColor,line width= 0.4pt,line join=round,line cap=round] ( 82.62, 61.20) -- ( 82.62, 55.20);

\path[draw=drawColor,line width= 0.4pt,line join=round,line cap=round] (123.13, 61.20) -- (123.13, 55.20);

\path[draw=drawColor,line width= 0.4pt,line join=round,line cap=round] (163.65, 61.20) -- (163.65, 55.20);

\path[draw=drawColor,line width= 0.4pt,line join=round,line cap=round] (204.17, 61.20) -- (204.17, 55.20);

\node[text=drawColor,anchor=base,inner sep=0pt, outer sep=0pt, scale=  1.00] at ( 82.62, 39.60) {100};

\node[text=drawColor,anchor=base,inner sep=0pt, outer sep=0pt, scale=  1.00] at (123.13, 39.60) {200};

\node[text=drawColor,anchor=base,inner sep=0pt, outer sep=0pt, scale=  1.00] at (163.65, 39.60) {300};

\node[text=drawColor,anchor=base,inner sep=0pt, outer sep=0pt, scale=  1.00] at (204.17, 39.60) {400};

\path[draw=drawColor,line width= 0.4pt,line join=round,line cap=round] ( 49.20, 67.01) -- ( 49.20,196.23);

\path[draw=drawColor,line width= 0.4pt,line join=round,line cap=round] ( 49.20, 67.01) -- ( 43.20, 67.01);

\path[draw=drawColor,line width= 0.4pt,line join=round,line cap=round] ( 49.20, 99.32) -- ( 43.20, 99.32);

\path[draw=drawColor,line width= 0.4pt,line join=round,line cap=round] ( 49.20,131.62) -- ( 43.20,131.62);

\path[draw=drawColor,line width= 0.4pt,line join=round,line cap=round] ( 49.20,163.93) -- ( 43.20,163.93);

\path[draw=drawColor,line width= 0.4pt,line join=round,line cap=round] ( 49.20,196.23) -- ( 43.20,196.23);

\node[text=drawColor,rotate= 90.00,anchor=base,inner sep=0pt, outer sep=0pt, scale=  1.00] at ( 34.80, 67.01) {0.4};

\node[text=drawColor,rotate= 90.00,anchor=base,inner sep=0pt, outer sep=0pt, scale=  1.00] at ( 34.80, 99.32) {0.6};

\node[text=drawColor,rotate= 90.00,anchor=base,inner sep=0pt, outer sep=0pt, scale=  1.00] at ( 34.80,131.62) {0.8};

\node[text=drawColor,rotate= 90.00,anchor=base,inner sep=0pt, outer sep=0pt, scale=  1.00] at ( 34.80,163.93) {1.0};

\node[text=drawColor,rotate= 90.00,anchor=base,inner sep=0pt, outer sep=0pt, scale=  1.00] at ( 34.80,196.23) {1.2};

\path[draw=drawColor,line width= 0.4pt,line join=round,line cap=round] ( 49.20, 61.20) --
	(242.20, 61.20) --
	(242.20,218.20) --
	( 49.20,218.20) --
	cycle;
\end{scope}
\begin{scope}
\path[clip] (  0.00,  0.00) rectangle (267.40,267.40);
\definecolor{drawColor}{RGB}{0,0,0}

\node[text=drawColor,anchor=base,inner sep=0pt, outer sep=0pt, scale=  1.00] at (145.70, 15.60) {$t$};

\end{scope}
\begin{scope}
\path[clip] ( 49.20, 61.20) rectangle (242.20,218.20);
\definecolor{drawColor}{RGB}{0,0,0}

\path[draw=drawColor,line width= 0.4pt,dash pattern=on 4pt off 4pt ,line join=round,line cap=round] (235.05, 72.51) --
	(200.85, 75.05) --
	(179.76, 76.01) --
	(164.21, 76.38) --
	(151.74, 76.53) --
	(141.23, 76.62) --
	(132.10, 76.75) --
	(123.96, 76.96) --
	(116.57, 77.28) --
	(109.77, 77.70) --
	(103.42, 78.22) --
	( 97.42, 78.82) --
	( 91.68, 79.45) --
	( 86.12, 80.02) --
	( 80.64, 80.39) --
	( 75.14, 80.49) --
	( 69.48, 80.62) --
	( 63.41, 81.75) --
	( 56.35, 87.14);

\path[draw=drawColor,line width= 0.4pt,dash pattern=on 4pt off 4pt ,line join=round,line cap=round] (235.05,136.17) --
	(200.85,130.12) --
	(179.76,128.39) --
	(164.21,128.35) --
	(151.74,129.17) --
	(141.23,130.50) --
	(132.10,132.18) --
	(123.96,134.14) --
	(116.57,136.39) --
	(109.77,138.98) --
	(103.42,142.03) --
	( 97.42,145.74) --
	( 91.68,150.44) --
	( 86.12,156.68) --
	( 80.64,165.30) --
	( 75.14,177.02) --
	( 69.48,191.84) --
	( 63.41,208.95) --
	( 56.35,224.56);
\definecolor{drawColor}{RGB}{147,112,219}

\path[draw=drawColor,line width= 0.4pt,line join=round,line cap=round] (254.88, 90.22) --
	(209.74, 93.62) --
	(183.72, 96.01) --
	(165.51, 97.91) --
	(151.53, 99.52) --
	(140.19,100.95) --
	(130.64,102.24) --
	(122.38,103.45) --
	(115.07,104.62) --
	(108.48,105.77) --
	(102.44,106.96) --
	( 96.80,108.26) --
	( 91.44,109.73) --
	( 86.29,111.40) --
	( 81.29,113.24) --
	( 76.36,115.02) --
	( 71.35,116.35) --
	( 65.88,117.13) --
	( 59.05,117.72);
\definecolor{drawColor}{RGB}{220,220,220}

\path[draw=drawColor,draw opacity=0.20,line width= 0.8pt,line join=round,line cap=round] ( 49.20,105.46) -- (242.20,105.46);
\end{scope}
\end{tikzpicture}
\end{subfigure}
\caption{Estimated $\theta_{\obs}$ (solid black) and corresponding $95$\% confidence intervals (dashed black) on the quantile scale (left) and time scale (right). Furthermore, $\theta$ for $S_{T^0}(t) = S_{T | A = 0}(t)$ and  $S_{T^{1}}(t) = 0.5 S_{T^{0}}(0.9t) + 0.5 S_{T^{0}}(0.45t)$ is presented (purple).}
\label{fig:case_study}
\end{figure}

Due to the small sample size, the $\theta_{\obs}$ curve is subject to substantial statistical uncertainty so that a constant acceleration factor cannot be ruled out. In the remaining discussion, we will ignore the statistical uncertainty. One might conclude that there is a time-varying acceleration factor such that the treatment becomes more beneficial over time, i.e. the $F_{T|A=1}(t)$ quantile of $T|A=0$ equals $t\,\theta_{\obs}(t)$ and thus decreases relative to $t$ over time. As explained in this paper, one cannot distinguish such a time-varying causal effect from a time-invariant heterogeneous causal effect. 

The trial conducted by Bouché et al. was originally a three-arm trial comparing LV5FU2 alone with LV5FU2 plus cisplatin and LV5FU2 plus irinotecan. In the data used for the GASTRIC meta-analysis, the two active treatment arms were merged, yielding a binary comparison between LV5FU2 alone and LV5FU2 plus either cisplatin or irinotecan. Thus, in our re-analysis, $A=0$ denotes LV5FU2 alone, while $A=1$ denotes the combined active-treatment group. 

In the case where cisplatin and irinotecan have different (but homogeneous) effects, there is treatment effect heterogeneity in the combined group. The presence of multiple versions of treatment violates causal consistency, but Theorem \ref{prop:1} also applies under a weaker form of causal consistency so that $T\mid A=a$ is equal in distribution to $T^{a}$. Such heterogeneity can result in a time-varying acceleration factor. For example, assume the acceleration factor for cisplatin is $0.9$ and for irinotecan is $0.45$. Then the survival function for receiving one of these treatments with probability $0.5$ equals $S_{T^{1}}(t) = 0.5 S_{T^{0}}(0.9t) + 0.5 S_{T^{0}}(0.45t)$.
When $S_{T^{0}}$ equals the $S_{T\mid A=0}$ distribution fitted before, the resulting $\theta$ is shown with the purple line in \cref{fig:case_study}. These time-invariant but heterogeneous effects could quite well explain the estimated time-varying acceleration factor.

\section{Confounding}\label{S:Confounding}
So far, we have supposed the absence of confounding. However, the presented results also hold for observational data when all confounders $L$ are observed, by conditioning on $L$. In particular, when $T^a \indep A | L$, then the conditional (on $L$) observed acceleration factor $\frac{1}{t} S_{T |A=0, L = \ell}^{-}(S_{T | A=1, L = \ell}(t)),$ identifies the conditional causal acceleration factor,  
\begin{align*}
      \theta_{L=\ell}(t) &\coloneqq \frac{1}{t} S_{T^0 | L = \ell}^{-}(S_{T^1 | L = \ell}(t)),
\end{align*}
cf. \cref{prop:1}. In turn, when the positivity assumption $(\forall \ell{:}~ 0 < \mathbb{P}(A~{=}~1  \mid L{=}\ell) < 1)$ applies, as discussed in Chapter 3 of Hern\'{a}n and Robins,\supercite{hernan2020} the causal acceleration factor $\theta$ (\cref{def:caf}) is equal to
\begin{align*}  \theta_\mathrm{\, adj} (t) \coloneqq \frac{1}{t} S^{-}_{0, \mathrm{\, adj}} \left( S_{1, \mathrm{\, adj}} (t) \right),
\end{align*}
where $S_{a, \mathrm{\, adj}}(t)  = \int S_{T| L = \ell, A=a} \,(t) d F_L(\ell) $. 


In the absence of effect heterogeneity, the conditional causal acceleration factor $\theta_L$ does not vary with $L$, and thus coincides with $\theta_{\mathrm{adj}}$ and $\theta$. In systems with effect heterogeneity, however, $\theta_L$ will generally vary across strata of $L$, either because the distribution of $U_{1}$ depends on $L$ or because the distribution of $U_{0}$ depends on $L$, in which case the distribution of $T^0$ also depends on $L$. The latter parallels the behavior illustrated in \cref{fig:effect_hetero}, where $\theta$ depends on the distribution of $T^0$ in the presence of $U_1$. Likewise, the distributions $T^0 \mid L=\ell$ yield different $\theta_L$ values when effect heterogeneity is present. Consequently, $\theta_L$ deviates from $\theta$ whenever effect heterogeneity exists.


To illustrate both the deviation of $\theta_{\obs}$ from $\theta$ due to confounding and the deviation of $\theta_L$ from $\theta$ in the presence of heterogeneity, we extend the settings from \cref{table:1} and \cref{fig:effect_hetero2} by introducing a measured confounder $L$. Specifically, the setup of \cref{table:1} illustrates the homogeneous-effect setting. Here $T^1 = T^0 / 3^{1/3}$, $T^0$ follows a Gamma–Weibull distribution ($\rho_0 = 1$), and we add a Bernoulli-distributed ($p=0.5$) confounder $L$ that causes $A$ and is associated with $U_0$. The setup of \cref{fig:effect_hetero2} demonstrates the heterogeneous-effect setting. Here $T^1 = T^0 / U_1$, $T^0$ follows a Gamma–Weibull distribution ($\rho_0 = 1$), and $U_1$ follows a Gamma distribution with $\mathrm{Var}(U_1) = 1$ and $\mathbb{E}[U_1] = 3^{1/3}$, and we add a Bernoulli-distributed ($p=0.5$) confounder $L$ that causes $A$ and is associated with either $U_0$ or $U_1$. 

The causal structures for the extended examples are presented in the single-world intervention graphs in \cref{fig:SWIGconf}. The homogeneous setting is presented in \cref{swig:label1} and the heterogeneous setting is presented in \cref{swig:label2} with the additional restriction that $U_{0} \indep U_{1}$ since this holds in our example (so either $(i)$ $U_0 \rightarrow  L \leftarrow U_1$, $(ii)$ $L \indep U_1$ and $U_0 \rightarrow L$ or $L \rightarrow U_0$ or $(iii)$ $L \indep U_0$ and $U_1 \rightarrow L$ or $L \rightarrow U_1$). In the SCMs defined by\cref{SCAFT-effect-homo} and \cref{SCAFT_lin_rel}, randomness in $T^a$ (apart from $N_T$) comes from $U_0$ or $U_1$, so dependence between $L$ and $T^a$ must be related to $U_0$ or $U_1$. 

\begin{figure}[htbp]
\centering
\definecolor{steelblue}{HTML}{4682b4}
\begin{subfigure}[b]{0.45\textwidth}
\centering
\resizebox{0.8\textwidth}{!}{
\begin{tikzpicture}
\tikzset{line width=1.5pt, outer sep=0pt,
ell/.style={draw,fill=white, inner sep=2pt,
line width=1.5pt},
swig vsplit={gap=5pt,
inner line width right=0.5pt}};
\tikzset{swig vsplit={gap=3pt, line color right=steelblue}}
\node[name=A,shape=swig vsplit]{
\nodepart{left}{$A$}
\nodepart[text=steelblue]{right}{$a$}};
\node[name=Ta, right=20mm of A, ell, shape=ellipse, minimum width=24pt, draw=steelblue, text=steelblue]{$T^a$};
\node[name=U0, above=10mm of Ta, xshift=-6mm, ell, shape=ellipse]{$U_0$};
\node[name=L, above=15mm of A, ell, shape=ellipse, minimum size=20pt]{$L$};

\draw[->,line width=1.5pt,>=stealth, draw=steelblue, fill=steelblue]
    (A) edge (Ta)
    ;
\draw[->,line width=1.5pt,>=stealth, bend left, looseness=0.5]
    (U0) edge (Ta)
    ;
\draw[<-,line width=1.5pt,>=stealth, bend left, looseness=0.8]
    (A) edge (L)
    ;
\draw[-,line width=1.5pt,>=stealth, dashed, bend left, looseness=0.3]
    (L) edge (U0)
    ;
\end{tikzpicture}
}
\caption{\label{swig:label1}}
\end{subfigure}
\begin{subfigure}[b]{0.45\textwidth}
\centering
\resizebox{0.8\textwidth}{!}{
\begin{tikzpicture}
\tikzset{line width=1.5pt, outer sep=0pt,
ell/.style={draw,fill=white, inner sep=2pt,
line width=1.5pt},
swig vsplit={gap=5pt,
inner line width right=0.5pt}};
\tikzset{swig vsplit={gap=3pt, line color right=steelblue}}
\node[name=A,shape=swig vsplit]{
\nodepart{left}{$A$}
\nodepart[text=steelblue]{right}{$a$}};
\node[name=Ta, right=20mm of A, ell, shape=ellipse, minimum width=24pt, draw=steelblue, text=steelblue]{$T^a$};
\node[name=U0, above=10mm of Ta, xshift=-6mm, ell, shape=ellipse]{$U_0$};
\node[name=U1, right=4mm of Ta, above=10mm of Ta, xshift=5mm, ell, shape=ellipse]{$U_1$};
\node[name=L, above=15mm of A, ell, shape=ellipse, minimum size=20pt]{$L$};

\draw[->,line width=1.5pt,>=stealth, draw=steelblue, fill=steelblue]
    (A) edge (Ta)
    ;
\draw[->,line width=1.5pt,>=stealth, bend right, looseness=0.6]
    (U0) edge (Ta)
    ;
\draw[->,line width=1.5pt,>=stealth, bend left, looseness=0.6]
    (U1) edge (Ta)
    ;
\draw[<-,line width=1.5pt,>=stealth, bend left, looseness=0.8]
    (A) edge (L)
    ;
\draw[-,line width=1.5pt,>=stealth, dashed, bend left, looseness=0.6]
    (L) edge (U1)
    ;
\draw[-,line width=1.5pt,>=stealth, dashed, bend left, looseness=0.2]
    (L) edge (U0)
    ;
\end{tikzpicture}
}
\caption{\label{swig:label2}}
\end{subfigure}
\caption{SWIGs for SCMs defined by \cref{SCAFT-effect-homo} and \cref{SCAFT_lin_rel} extended with a (potential) confounder $L$. In \cref{swig:label1}, the undirected dashed line indicates the presence or absence of arrows in either direction. In \cref{swig:label2}, the undirected dashed lines indicate the presence or absence of arrows in either direction under the restriction $U_0 \indep U_1$, i.e. the three possible scenarios: ($i$) $U_0 \rightarrow L \leftarrow U_1$, ($ii$) $L \indep U_1$ and $U_0 \rightarrow L$ or $L \rightarrow U_0$ and ($iii$) $L \indep U_0$ and $U_1 \rightarrow L$ or $L \rightarrow U_1$.}
\label{fig:SWIGconf}
\end{figure}
The marginals of $(L, U_{0}, U_1)$ are generated using a Gaussian copula so that Kendall's $\tau$ correlations of $L$ with $U_{0}$ and $U_{1}$ equal $\tau_{0}$ and $\tau_{1}$ respectively. Given $L=\ell$, $\mathbb{P}(A=1\mid L=\ell)= 0.5 + \beta_{\text{LA}} (2\ell -1)$, so that Kendall's $\tau$ for $L$ and $A$ equals $2\beta_{\text{LA}}$. The code for this simulation can be found at \url{https://github.com/marbrath/causal_AFT}.  

For the homogeneous-effect setting of \cref{swig:label1}, the quantities $\theta_{\obs}$, $\theta_{\mathrm{adj}}$, and $\theta_L$ were obtained empirically from simulations with $n_{\mathrm{obs}} = 10^6$ individuals, and are presented for $\tau_0 \in \{0, 0.5\}$ in \cref{fig:conf-homo}, where $\theta$ is also shown. When there is no confounding ($\tau_0 = 0$), all estimands coincide. In contrast, when confounding is present ($\tau_0 = 0.5$), $\theta_{\obs}$ deviates from $\theta$, whereas both $\theta_{\mathrm{adj}}$ and the stratum-specific effects $\theta_{L}$ identify $\theta$. The magnitude of the deviation of $\theta_{\obs}$ increases with larger $\beta_{LA}$ as illustrated in Figure S1.

\begin{figure}[htbp]
\centering
\includegraphics[width=0.8\linewidth]{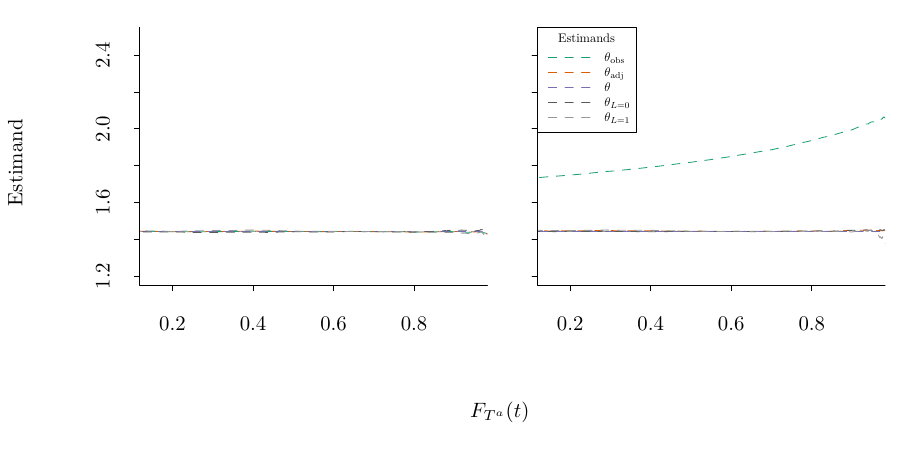}
\caption{$\theta_{\mathrm{adj}}$ (orange), $\theta_{\obs}$ (green), $\theta$ (purple), $\theta_{L=0}$ (dark gray), and $\theta_{L=1}$ (light gray) for the setting in \cref{swig:label1}. Here, $T^1 = T^0 / 3^{1/3}$, $T^0$ follows a Gamma–Weibull distribution ($\rho_0 = 1$). Furthermore, $L \sim \mathrm{Bernoulli}(0.5)$, $\beta_{LA} = 0.25$, with $\tau_0 = 0$ (left) and $\tau_0 = 0.5$ (right). See the associated $S_{T^0}$ and $S_{T^a}$ in Figure S6.}
\label{fig:conf-homo}
\end{figure}

For the heterogeneous setting of \cref{swig:label2}, the quantities $\theta_{\obs}$, $\theta_{\mathrm{adj}}$, and $\theta_L$ were obtained empirically from simulations with $n_{\mathrm{obs}} = 10^6$ individuals, and are presented for $\tau_0, \tau_1 \in \{0,0.5\}$ in \cref{fig:conf-hetero}, where $\theta$ is also shown. In contrast to \cref{fig:conf-homo}, the stratum effects $\theta_L$ differ from $\theta$, which is only identified via $\theta_{\mathrm{adj}}$. For the example considered here, confounding resulting from a relation between $L$ and $U_{0}$ produces a larger deviation of $\theta_{\obs}$ from $\theta$ than does confounding resulting from a relation between $L$ and $U_1$. When both relations are present, the deviation is even larger. The larger $\beta_{\text{LA}}$, the larger the deviation of $\theta_{\obs}$ from $\theta$ as illustrated in Figures S2-S4 in the Supplementary Material.

\begin{figure}[htbp]
\centering
\includegraphics[width=\linewidth]{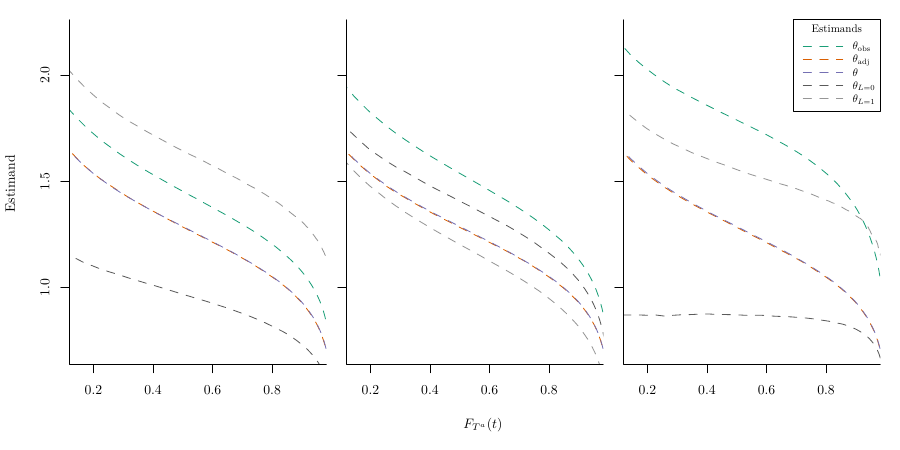}
\caption{$\theta_{\mathrm{adj}}$ (orange), $\theta_{\obs}$ (green), $\theta$ (purple), $\theta_{L=0}$ (dark gray), and $\theta_{L=1}$ (light gray) for the setting in \cref{swig:label2}. Here, $T^1 = T^0 / U_1$, $T^0$ follows a Gamma–Weibull distribution ($\rho_0 = 1$), and $U_1$ follows a Gamma distribution with $\mathrm{Var}(U_1) = 1$ and $\mathbb{E}[U_1] = 3^{1/3}$. Furthermore, $L \sim \mathrm{Bernoulli}(0.5)$, $\beta_{LA} = 0.25$, and $(\tau_0, \tau_1) = (0.5, 0)$ (left), $(0, 0.5)$ (middle), and $(0.5, 0.5)$ (right). See the associated $S_{T^0}$ and $S_{T^a}$ in Figure S10.}
\label{fig:conf-hetero}
\end{figure}

\newpage
\section{Discussion}
In this work, we have formalized the causal interpretation of the acceleration factor estimand $\theta(t)$ in AFT models, and we have shown that it yields an appropriate causal effect measure in the presence of frailty and treatment effect heterogeneity. When the acceleration factor is estimated in an unrestricted, or saturated, AFT formulation, meaning that $\theta(t)$ is allowed to vary with time without any imposed functional form, the estimated acceleration factor can be used directly to answer the corresponding causal question. In restricted models, for example, models assuming a time-invariant acceleration factor, this causal interpretation is retained only if the imposed restriction is well-specified.

The acceleration factor, which may in general vary with time, admits a clear causal interpretation as the ratio of quantiles of the potential outcome distributions under both frailty and effect heterogeneity. In contrast, the parameter of a proportional hazards model does not retain a clear causal interpretation in the presence of either form of heterogeneity, and one should additionally derive the survival curves corresponding to the fitted model and use a measure that contrasts the survival curves to quantify the effect. 

In the absence of effect heterogeneity, the conditional and marginal acceleration factors coincide. Hence, the acceleration factor has both a conditional interpretation, as the ratio of quantiles of the conditional potential outcome distributions given $U_0$, and a marginal, population-level interpretation, as the ratio of quantiles of the marginal potential outcome distributions. In contrast, the observed hazard ratio generally differs from the corresponding conditional causal hazard ratio in the presence of frailty, due to selection among survivors.

When effect heterogeneity is present, the conditional and marginal versions of both the acceleration factor and hazard ratio differ, and the marginal estimands admit only population-level interpretations. Consequently, as demonstrated in \cref{sec:homo} and \cref{subsec:effect_hetero}, both marginal estimands may vary over time, even though the corresponding conditional estimands remain time-invariant. This time variation further complicates the causal interpretation of the hazard ratio: reasoning about time variation in an estimand which conflates causal effects with selection is generally uninformative. For the acceleration factor, however, the estimand itself retains a transparent causal interpretation, and its time variation can be understood directly as the evolving ratio of quantiles of a mixture of potential outcome distributions.

Because effect heterogeneity induces time variation in the marginal acceleration factor, a time-invariant AFT model is necessarily misspecified in such settings. Table S1 in the Supplementary Material displays empirically obtained estimands $\mathbb{E}[T^0] / \mathbb{E}[T^1]$ and $\exp \left( \mathbb{E}[\log T^0] - \mathbb{E}[\log T^1] \right)$ for all examples presented in this paper. This demonstrates that for misspecified time-invariant AFTs (i.e. when $U_1$ is present) the estimands of time-invariant AFTs given in \cref{eq:est1} and \cref{eq:est2} cannot be viewed as simple and meaningful summary measures of the treatment effect and will depend on the follow-up times. Thus, in the presence of effect heterogeneity, fitting the misspecified time-invariant AFT model to two studies with different follow-up times targets two different estimands so that the results are not comparable. By contrast, the time-varying acceleration factor $\theta_{\obs}(t)$ is the same for both studies until the earlier of the two study end times and is simply not defined for one of the studies thereafter. Consequently, valid inference on the acceleration factor requires time-varying AFT models whenever effect heterogeneity is believed to be present.

To account for confounding, AFT models may be fitted using inverse probability of treatment weighting and also offer an intuitive framework for modeling time-to-event data with covariate adjustment. As shown in Section \ref{S:Confounding}, the presented results generalize to settings with confounding, provided that there is no unmeasured confounding, because conditional AFT models can be used to estimate a standardized acceleration factor that identifies $\theta$. However, in the presence of confounding, $\theta_{\obs}$ deviates from $\theta$.  Additionally, in the presence of effect heterogeneity, $\theta_{L}$ may vary across strata of all variables. As a consequence, stratum-specific acceleration factors should not be interpreted as identifying the marginal causal effect. Practitioners must therefore carefully distinguish between conditional and marginal interpretations when adjusting for confounders and explicitly justify the absence of unmeasured confounding in order to support causal conclusions based on AFT models. In a setting with time-varying treatments and treatment-confounder feedback, which is out of scope in this work, g-estimation can be used to appropriately adjust for time-varying confounding. For example, Hern\'{a}n et al.\supercite{hernan2005structural} introduce simple structural (i.e., causal) AFTs for potential time-to-event outcomes under different levels of an intervention, although these models do not include effect heterogeneity.

Our results should not be interpreted as prescribing the use of acceleration factors rather than alternative causal estimands such as survival probabilities or restricted mean survival time. Rather, they clarify that, under correct model specification, the parameter of an AFT model is itself a causal estimand with a clear interpretation. In contrast, when interest ultimately lies in survival probabilities or restricted mean survival time, these quantities typically need to be obtained through additional steps involving marginalization over the covariate distribution. Likewise, proportional hazards models require the survival function to be derived from the fitted hazard model, as the hazard ratio itself generally lacks a straightforward causal interpretation. AFT models nevertheless still permit the subsequent derivation of survival curves and restricted mean survival times when these measures are of primary interest.

In summary, we have demonstrated that AFT models provide a coherent framework for causal inference in time-to-event settings because the acceleration factor retains a clear causal interpretation. 
When effect heterogeneity is present, valid inference on the acceleration factor requires time-varying AFT models. Their practical implementation may still present a significant hurdle for practitioners, making this an important direction for future methodological research.

\section*{Acknowledgements}
This work was supported by the South-Eastern Norway Regional Health Authority (Grant No. $2019007$).

The authors thank the GASTRIC (Global Advanced/Adjuvant Stomach Tumor Research International Collaboration) Group for permission to use the data analyzed in Section \ref{S:casestudy}. The investigators who contributed to the GASTRIC collaboration are listed elsewhere.\supercite{gastric2010benefit,Gastric2013}

OpenAI ChatGPT and Codex were used during manuscript revision to assist
with language editing, \LaTeX{} and figure formatting. These tools
were not used to generate or analyze data or to develop the scientific
methods, results, or conclusions. All outputs and resulting changes were
reviewed by the authors, who take full responsibility for the final
content.

\appendix
\section{Proofs}
\subsection{Proof of \Cref{prop:1}}
\begin{proof}
The result follows by exchangeability and causal consistency, 
\begin{align*}
    S_{T^a} (t) &= \mathbb{P}(T^a > t) = \mathbb{P}(T^a > t | A = a) = \mathbb{P}(T > t | A = a) = S_{T | A = a }(t ).
\end{align*}
\end{proof}

\subsection{Proof of \Cref{prop:kaplan-meier}}
\begin{proof}
Follows by \cref{prop:1} and the fact that the observed survival functions are identified from censored data under the independent censoring assumption.
\end{proof}

\clearpage
\printbibliography

\clearpage
\begin{table}[h]
\centering
\begin{subtable}{\textwidth}
\centering
\begin{tabular}{cc|cc|cc}
$U_0 \sim$ & $\rho_0$ & $\theta$ & $\mathbb{E}[\hat{\theta}_{\obs}]$ & $\exp \beta$ & $ \mathbb{E}[\exp \hat{\gamma}] $\\[4pt]
\hline
     & 0.5 & $0.693$ & $0.693 \, (0.693, 0.694)$ & 1/3 & $0.477 \, (0.467, 0.488)$ \\
$\Gamma(\rho^{-1}_0, \rho_0)$
     & 1 & $0.693$ & $0.693 \, (0.693, 0.694)$ & 1/3 & $0.577 \, (0.558, 0.598)$ \\
     & 2 & $0.693$ & $0.693 \, (0.693, 0.694)$ & 1/3 & $0.695 \, (0.663, 0.728)$ \\
\hline
     & 0.5 & $0.693$ & $0.693 \, (0.693, 0.694)$ & 1/3 & $0.421 \, (0.415, 0.427)$ \\
$\mathrm{IG}(1, \rho^{-1}_0)$
     & 1 & $0.693$ & $0.693 \, (0.693, 0.694)$ & 1/3 & $0.459 \, (0.450, 0.468)$ \\
     & 2 & $0.693$ & $0.693 \, (0.693, 0.694)$ & 1/3 & $0.495 \, (0.483, 0.508)$ \\
\end{tabular}
\caption{$\beta = \log(1/3)$ \label{tbl:beta_1}}
\end{subtable}
\begin{subtable}{\textwidth}
\centering
\begin{tabular}{cc|cc|cc}
$U_0 \sim$ & $\rho_0$ & $\theta$ & $\mathbb{E}[\hat{\theta}_{\obs}]$ & $\exp \beta$ & $ \mathbb{E}[\exp \hat{\gamma}] $\\[4pt]
\hline
     & 0.5 & $1.442$ & $1.442 \, (1.442, 1.443)$ & 3 & $2.095 \, (2.048, 2.143)$ \\
$\Gamma(\rho^{-1}_0, \rho_0)$
     & 1 & $1.442$ & $1.442 \, (1.442, 1.443)$ & 3 & $1.732 \, (1.673, 1.793)$ \\
     & 2 & $1.442$ & $1.442 \, (1.442, 1.443)$ & 3 & $1.439 \, (1.374, 1.508)$ \\
\hline
     & 0.5 & $1.442$ & $1.442 \, (1.442, 1.442)$ & 3 & $2.375 \, (2.340, 2.411)$ \\
$\mathrm{IG}(1, \rho^{-1}_0)$
     & 1 & $1.442$ & $1.442 \, (1.442, 1.443)$ & 3 & $2.181 \, (2.137, 2.225)$ \\
     & 2 & $1.442$ & $1.442 \, (1.442, 1.443)$ & 3 & $2.020 \, (1.970, 2.071)$ \\
\end{tabular}
\caption{$\beta = \log(3)$ \label{tbl:beta_2}}
\end{subtable}
\caption{Empirical mean and corresponding $95\%$ confidence interval of the estimators obtained by fitting marginal AFT and Cox proportional hazards models to observations $T\mid A$ generated under the true conditional Weibull proportional hazards model
$\lambda_i^a(t) = (t^2/20) U_{0i} \exp(\beta a),$
with $\beta=\log(\theta^3)$, $n_{\mathrm{obs}}=10^5$, and $n_{\mathrm{sim}}=1000$. The Cox model fitted to the marginal data is denoted
$\lambda(t\mid A)=\lambda_0(t)\exp(\gamma A),$
so the final column reports $\mathbb{E}[\exp(\hat{\gamma})]$. See associated $S_{T^0}$ and $S_{T^1}$ in Figure S6.}
\label{table:1}
\end{table}

\end{document}